\documentclass[11pt]{article}
\usepackage{amsmath,amsthm,amsfonts,authblk}
\usepackage{amssymb,graphicx,verbatim,multirow}
\usepackage[round]{natbib}
\usepackage{color}

{\bf}
{\bf}

\def\d{\delta}

\def\d{\delta}

\newtheorem{theorem}{Theorem}%[subsection]
\newtheorem{proposition}[theorem]{Proposition}%[subsection]
\title{Piecewise linear models of chemical reaction networks}
\author[a,c]{Ajit Kumar\footnote{Correspondence Author \\
Ajit Kumar (\texttt{ajit.kumar@snu.edu.in}), Kre\v{s}imir Josi\'{c} (\texttt{josic@math.uh.edu}) }}

\author[a,b]{Kre\v{s}imir Josi\'{c}}

\affil[a]{Department of Mathematics, University of Houston, Houston, Texas 77204-3008, USA}
\affil[b]{Department of Biology and Biochemistry, University of Houston, Houston, Texas 77204-5001, USA}
\affil[c]{Department of Mathematics, Shiv Nadar Univerisity, Greater Noida, Uttar Pradesh 203207, India}

\begin{document}
\maketitle
\abstract{We show that certain non-linear dynamical systems
with non-linearities in the form of Hill functions, can
be approximated by piecewise linear dynamical systems.  The resulting
piecewise systems have closed form solutions that can be used to 
understand the behavior of the fully nonlinear system.  We justify the reduction using
geometric singular perturbation theory, and illustrate the results in networks modeling a genetic switch and a genetic oscillator.}
%\chapter{Linearization of Michaelis-Menten differential
%equations}\label{chapter:LinearMM}

%\Kcomment{We now use both domain, subdomain and region, and subregions.  I
%think we should stick
%with domain and subdomain.  Could you change this so it is consistent.}

\section{Introduction}

Accurately describing the behavior of interacting enzymes, proteins, and genes
requires spatially extended stochastic models.  However, such models are
difficult to implement and fit to data, hence modelers frequently  use tractable reduced models. 
In most popular models of biological networks, the 
dynamics of each node  is described by a single
ODE, and sigmoidal functions are used to model interactions between the network
elements.  The resulting ODEs are generally not analytically tractable.
This can hinder the study of  large networks, where the number of parameters and
 the potential dynamical complexity make it difficult to analyze the behavior of the system using
purely
 numerical methods.

Analytical treatments are possible in certain limits.  For instance, the
approaches that have been developed  to analyze models of gene interaction
networks can be broadly classified into three
categories~[\cite{PolynikisHoganBernardo2009}]:  \emph{Quasi Steady State Approximations} (QSSA), \emph{Piecewise Linear Approximations} (PLA), and
 \emph{discretization of continuous time
ODEs}.

Here, we aim to develop the theory of PLAs.  In
certain limits interactions between network elements become
switch--like~[\cite{kauffman69,Alon2006,DavitichBornholdt2008}].   For instance,
the Hill function,   $f(x) = x^n / (x^n + J^n )$,  approaches the Heaviside function, $H(x - J)$,
in the limit of large $n$.
In this limit the domain on which the network is modeled
is also naturally broken into subdomains.  For Hill functions, the thresholds,
defined by $J$, divides the domain into two subdomains within which
the Heaviside function is constant.  Thus within each subdomain a node is either
fully expressed, or
not expressed at all.  The original Hill function, $f(x)$, is approximately constant in each
of the subdomains, and boundary layers occur when $x$ is
close to threshold~[\cite{IroniPanzeriPlahteSimoncini2011}].

This general approach has a long and rich history, and piecewise linear
functions of the form proposed in~[\cite{GlassKauffman1973}] have been shown to
be well suited for the modeling of genetic regulatory networks (for a brief
review see~[\cite{de02}]).  In certain cases
the results can be justified rigorously.  In particular,  singular perturbation
theory  can
be used to obtain reduced equations within each subdomain and the boundary
layers,
and global approximations   within the entire
domain~[\cite{IroniPanzeriPlahteSimoncini2011}].

Here we take a similar approach, but work in a different limit.  We again start
with the Hill function, $ x^n / (x^n + J^n )$, but assume that $J$ is small. 
Although the subsequent results hold for any fixed $n$, for simplicity we
assume $n = 1$.
%Thus the network interactions are modeled by
%functions of the form $ x / (x + J )$ (activation) or  $(1-x )/ (1- x + J )$
%(repression).
 Equations involving this special class of Hill functions are known as
Michaelis-Menten equations, and $J$ is known as the Michaelis-Menten 
constant~[\cite{MichaelisMenten1913,GoldbeterKoshland1981,
CilibertoFabrizioTyson2007,ChaoTang2009,DavitichBornholdt2008,Goldbeter1991,NovakTyson1993,
NovakPatakiCilibertoTyson2001,NovakPatakiCilibertoTyson2003}].  
We note that the models of chemical reactions we consider can be rigorously derived from
the Chemical Master Equation only in the case of a single reaction~[\cite{KumarJosic2011}].
The models of networks of chemical reactions that we take as the starting point of
our reduction should therefore be regarded as phenomenological.

We will examine the case when the Michaelis-Menten constant, $J$, is small. This case
 has a simple physical
interpretation: Consider the Hill function that occurs in the Michaelis-Menten
scheme, where an enzyme is catalyzing the conversion of the inactive form of some
protein to its active form. When $J$ is small
the total enzyme concentration is much smaller than the total protein
concentration.
The asymptotic limit $J \rightarrow 0$ was recently considered to obtain 
heuristically
a Boolean approximation of a protein interaction
network~[\cite{DavitichBornholdt2008}].  Here we consider a rigorous
justification underlying such reductions, as well as how the reduction could
be used to understand the dynamics of gene networks.

The main idea behind the reduction we propose can be summarized as follows:
Given the non-linear term $f(x) = x / (x + J )$, when $x
\gg J $ then $f(x) \approx 1$, and when $x \approx 0$ then  we do the analysis by introducing a new variable like $\tilde{x} := J/x$. This new variable $\tilde{x}$ serves as a microscope to observe the boundary regions.  As we will show, the domain is naturally decomposed  into a nested sequence 
of hypercubes such that for each level of
nesting we get a separate linear equation. 
%
%{\bf The previous needs more explanation.  Actually, I think you should have a
%figure that illustrates this
%in two dimensions explicitly showing where the different hypercubes are, and
%connecting it with the subsequent
%table.}

%\AK{I think the example in the toggle switch should be sufficient.}

We proceed as follows: In Section~\ref{ExampleProblems} we illustrate our
approach using simple examples and provide numerical evidence for the validity
of our claim. In Section~\ref{GeneralTheory} we describe a general class of
differential equations which subsumes these examples.
% \Kcomment{what do you mean by ``fit it''?} \AK{I meant ``fit in``.}
 Furthermore, in this section we justify our approach mathematically using
Geometric Singular Perturbation Theory (GSPT). We will conclude with a
discussion on limitations of these reductions.

\section{Example problems}\label{ExampleProblems}

We start by demonstrating the main idea of our approach in the cases
of two and three mutually repressing biological elements. For instance, these elements
could be genes that mutually inhibit each other's
production~[\cite{GardnerCantorCollins2000,ElowitzLeibler2000}].  
However, as
the theory we develop is general, we do not constrain it to a particular interpretation.
%{\AK{ I am not sure what do you mean by this sentence}}.
We first provide an intuitive illustration of the approach along with a
heuristic justification of the different steps in the
reduction. A  mathematical justification follows.

\begin{figure}[t]
\begin{center}
\includegraphics[scale = .5]{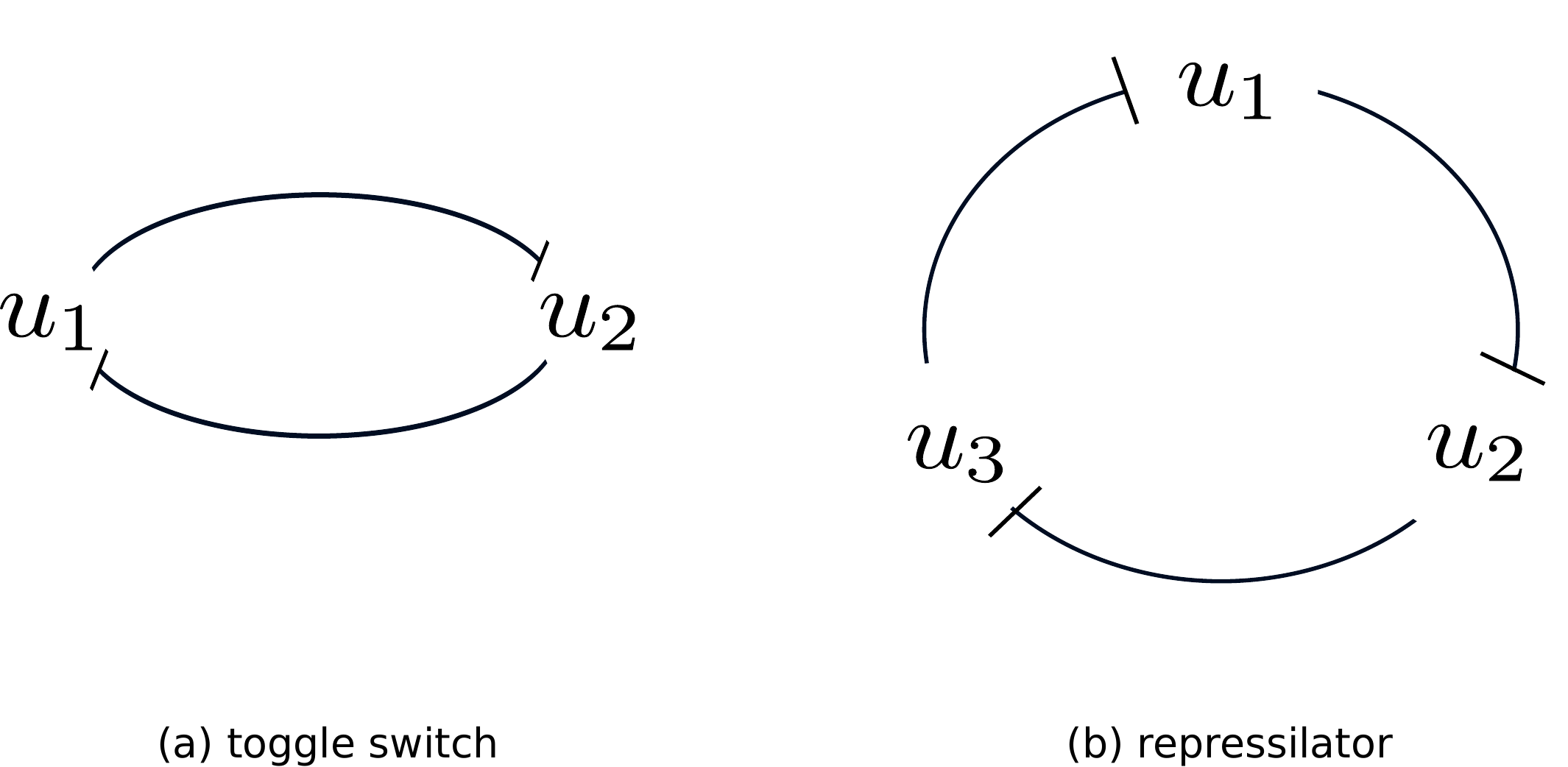}
\parbox{.9\textwidth}{
\caption{\footnotesize {(a) Nodes $u_1$, $u_2$ inhibiting each others activity. End result is like a switch. The node which was stronger in the beginning  will stay stronger and will completely suppress the other. (b) Nodes $u_1$, $u_2$, and $u_3$ suppressing each other in a cyclic fashion. Not surprisingly, the end result is oscillatory behavior.}%
\label{fig:toggle_repressillator}
}}
\end{center}
\end{figure}

\subsection{A network of two mutually inhibiting elements}

% \Kcomment{Change first sentence and refer to figure.}
First we consider two mutually repressing elements within a biological network. 
This  \emph{toggle switch}  motif (see Figure~\ref{fig:toggle_repressillator}a) 
is common in biological
networks~[\cite{NovakPatakiCilibertoTyson2003,GardnerCantorCollins2000}]. Let
$u_1,u_2 \in [0,1]$ represent the normalized levels of activity of the first and
second element, respectively. Therefore, when $u_i = 1$ the $i^{\text{th}}$
network element is maximally active (expressed). The system be modeled by
\begin{align}\label{EquationForToggleSwitch}
\begin{split}
\frac{du_1}{dt} &= 0.5\frac{1-u_1}{J+1-u_1}-u_2\frac{u_1}{J+u_1},  \\
\frac{du_2}{dt} &= 0.5\frac{1-u_2}{J+1-u_2}-u_1\frac{u_2}{J +u_2},
\end{split}
\end{align}
where $J$ is some positive constant. The structure of
Eq.~(\ref{EquationForToggleSwitch}) implies that the cube $[0,1]^2 = \{(u_1, u_2)
\,|\, 0 \le u_1, u_2 \le 1 \}$ is invariant (see
Proposition~\ref{prop:invariantCube}).

In the limit of small $J$, Eq.~(\ref{EquationForToggleSwitch})  can be
approximated by a piecewise linear differential equation as follows:
If $u_i$ is not too close to zero the expression $u_i/(J+u_i)$ is approximately unity.
  More precisely, we fix a small $\d > 0 $, which will be chosen to depend on
$J$.  When $u_i > \d$ and $J$ is small then $u_i/(J+u_i) \approx1$. Similarly, when 
$u_i > 1-\d$ then  $(1-u_i)/(J+1-u_i) \approx 1$.

With this convention in mind we  break the cube $[0,1]^2$ into several
subdomains, and define
a different reduction of Eq.~\eqref{EquationForToggleSwitch} within each.   For
example,  the interior of the domain $[0,1]^2$ is defined by
\begin{align}\label{R00}
\mathcal{R}^0_0:= \{(u_1,u_2) \in [0,1]^2 \,|\, \d \le u_1 \le 1-\d \text{ and }
\d \le u_2 \le 1-\d\}.
\end{align}
 Eq.~(\ref{EquationForToggleSwitch}), restricted to $\mathcal{R}^0_0$ is
approximated by the linear differential equation
\begin{align}\label{interiorLinear2d}
\frac{du_1}{dt} = 0.5-u_2, \quad
\frac{du_2}{dt} = 0.5-u_1.
\end{align}
%Eq.~(\ref{interiorLinear2d}) is obtained by replacing all the  rational
%expressions $u_1/ (J + u_1), (1-u_1)/ (J +1 - u_1) $, etc. by $1$
%\Kcomment{You don't need to repeat this.} .
%
On the other hand, if one of the coordinate is near the boundary, while the
other is in the interior, the approximation is different. For instance, the
region
\begin{align}\label{R10}
\mathcal{R}_1^0:= \{(u_1,u_2) \in [0,1]^2 \,|\,  u_1 < \d \text{ and } \d \le
u_2 \le 1-\d\},
\end{align}
forms a boundary layer where $u_1$ is of the same order as $J$. Therefore 
the term $u_1/(J + u_1)$ can not be approximated by unity.  Instead the approximation takes the form
\begin{subequations}\label{edgeLinear2d}
\begin{align}
\frac{du_1}{dt} &= 0.5-u_2 \frac{u_1}{J + u_1}, \label{edgeLinear2da} \\
\frac{du_2}{dt} &= 0.5-u_1. \label{edgeLinear2db}
\end{align}
\end{subequations}
This equation can be simplified further.  Since the boundary defined by $u_1= 
0$ is invariant,  $\frac{d u_1}{dt}$ must be small inside the boundary layer
$\mathcal{R}_1^0$.  We therefore use the approximations $\frac{d u_1}{dt}
\approx 0$ in Eq.~(\ref{edgeLinear2da}) and $u_1 \approx 0$ in 
Eq.~(\ref{edgeLinear2db}) to obtain
\begin{subequations}\label{smallu}
\begin{align}
0 &= 0.5-u_2 \frac{u_1}{J + u_1}, \label{smallua}\\
\frac{du_2}{dt} &= 0.5. \label{smallub}
\end{align}
\end{subequations}
Note that Eq.~(\ref{smallub}) is linear and decoupled from Eq.~(\ref{smallua}),
while Eq.~(\ref{smallua}) is an algebraic system which can be solved to obtain
$u_1 \approx J/(2u_2 - 1)$.
Within $\mathcal{R}_1^0$ we thus obtain the approximation $u_2(t) \approx 0.5 t 
+ u_2(0)$ and
$u_1(t) \approx J/( t + 2 u_2(0) - 1)$.

Note that here we have the freedom of only specifying
the initial condition $u_2(0)$, while $u_1(0)$ is determined from the solution
of the algebraic equation~\eqref{smallua}.  As we explain below, this algebraic
equation defines a slow manifold within the subdomain $\mathcal{R}_1^0$.  The
reduction assumes that solutions are instantaneously attracted to this manifold.
 %
%\begin{align}\label{smallu1Algebraic}
%\frac{du_2}{dt} = 0.5, \quad u_1 = J \frac{0.5}{u_2 - 0.5}.
%\end{align}
%

Table~\ref{table:2d} shows how these ideas can be extended to all of $[0,1]^2$. 
In each of the 9 listed subdomain one or both variables are close to either 0 or
1.
Therefore each subdomain corresponds to either the interior, edge, or corner of
the unit square. Following the preceding arguments, we assume that variable(s)
that are close to 0 or 1 are
  in steady state and lead to an algebraic equation.  Similarly, the evolution
of the interior variables is described by linear differential equations.
The resulting algebraic-differential systems are given in the last column of
Table~\ref{table:2d}.

The reductions in the corner subdomains 
$\mathcal{R}_{1,2}^0,\mathcal{R}_0^{1,2}, \mathcal{R}_1^2,$
and $\mathcal{R}_2^1$ consist of purely algebraic equations.  When $J$ is small
some of these
equations will have a solution in $[0,1]^2$, indicating a stable fixed point
near the corresponding corner ($\mathcal{R}_1^2$ and $\mathcal{R}_2^1$).  Others
will not have a solution in $[0,1]^2$, indicating that approximate solutions do
not enter the corresponding subdomain ($\mathcal{R}_{1,2}^0$ and
$\mathcal{R}_0^{1,2}$).

 Each approximate solution has the potential of exiting the subdomain within
which it is defined, and entering another.
The global approximate solution of Eq.~(\ref{EquationForToggleSwitch}) is 
obtained by using the exit point from one subdomain as the initial condition for
the approximation in the next.  In
subdomains other than $\mathcal{R}_0^0$ some of the initial conditions will be
prescribed by the algebraic part of the reduced system.  The global
approximation may therefore be discontinuous, as solutions entering a new
subdomain are assumed to instantaneously jump to the slow manifold defined by
the algebraic
part of the reduced system.
 Fig.~\ref{fig:twoDimension} shows that when $J$ is small, this approach
provides a good approximation.

\begin{table}[t]
$$
\begin{array}{c|c|c|rcl}
\hline
\text{Subdomain's name}		&  u_1 	& 	u_2	&
\multicolumn{3}{c}{\text{Approximating linear system}}  \\
\hline
\multirow{2}{*}{$\mathcal{R}_0^0$}	&  \multirow{2}{*}{$ \d \le u_1 \le 1-\d
$}	&  \multirow{2}{*}{ $ \d \le u_2 \le 1-\d $}	&   	 u_1'	&=&
0.5-u_2,			\\
								&		
					&					
  		 &    u_2'		&=& 0.5-u_1			\\
\hline
\multirow{2}{*}{$\mathcal{R}_0^1$}	& \multirow{2}{*}{$ u_1 > 1-\d $} 	
& \multirow{2}{*}{$ \d \le u_2 \le 1-\d $}		&  	           	
  0	&=&  \displaystyle 0.5\frac{1-u_1}{J+1-u_1}-u_2,	\\
								&		
					&					
		 &  	 u_2'  &=& -0.5				       \\
\hline
\multirow{2}{*}{$ \mathcal{R}_0^2$}     & \multirow{2}{*}{$ \d \le u_1 \le 1-\d
$}		& \multirow{2}{*}{$ u_2 > 1-\d $}	&  u_1'	&=&	 -0.5,	
				\\
								&		
					&					
	  	 &                         0	&=&	\displaystyle
0.5\frac{1-u_2}{J+1-u_2}-u_1	\\
\hline
\multirow{2}{*}{$ \mathcal{R}_1^0 $}	& \multirow{2}{*}{$ u_1 < \d $} 
& \multirow{2}{*}{$ \d \le u_2 \le 1-\d $}    	& 			
0 	 &=& \displaystyle 0.5-u_2\frac{u_1}{J+u_1},  \\
								&		
					&					
	   	 & 		u_2'&=& 0.5  \\
\hline
\multirow{2}{*}{$ \mathcal{R}_2^0	$}	& \multirow{2}{*}{$ \d \le u_1
\le 1-\d $}		&   \multirow{2}{*}{$u_2 <  \d $}	& 	u_1'	
&=& 0.5 ,  \\
								&		
					&					
		 &			0		&=& \displaystyle 0.5
-u_1\frac{u_2}{J +u_2} \\
\hline
\multirow{2}{*}{$ \mathcal{R}_0^{12}$}& \multirow{2}{*}{$ u_1 > 1-\d$} 	& 
\multirow{2}{*}{$u_2 > 1-\d$} 	&	0	&=& \displaystyle
0.5\frac{1-u_1}{J+1-u_1}-1,	\\
								&		
					&					
		 &	0	&=& \displaystyle 0.5\frac{1-u_2}{J+1-u_2}-1
\\
\hline
\multirow{2}{*}{$ \mathcal{R}^0_{12}$}&  \multirow{2}{*}{$u_1 < \d$}	 & 
\multirow{2}{*}{$u_2 < \d$}	&	0	&=& \displaystyle 0.5- J
\frac{u_1}{J+u_1},		\\
								&		
					&					
		 &	0	&=& \displaystyle 0.5- J\frac{u_2}{J +u_2}	
\\
\hline
\multirow{2}{*}{$ \mathcal{R}_2^1	$}	& \multirow{2}{*}{$ u_1 > 1-\d$}
	& \multirow{2}{*}{$u_2 < \d $}		&	0	&=&
\displaystyle  0.5\frac{1-u_1}{J+1-u_1},	\\
								&		
					&					
		 &	0	&=& \displaystyle 0.5- \frac{u_2}{J +u_2}	
\\
\hline
\multirow{2}{*}{$ \mathcal{R}_1^2	$}	& \multirow{2}{*}{$ u_1 < \d $}
&\multirow{2}{*}{$  u_2 > 1-\d$}		&	0	&=&
\displaystyle 0.5-\frac{u_1}{J+u_1},		\\
								&		
					&					
		 &	0	&=& \displaystyle 0.5\frac{1-u_2}{J+1-u_2}
\\
\hline
\end{array}
$$
\begin{center}
\parbox{.8\textwidth}{
\caption{\footnotesize  List of differential--algebraic systems that approximate
Eq.~\eqref{EquationForToggleSwitch} in different parts of the domain. The
subdomains are named so that the superscript (subscript)  lists the coordinates
that are close to $1$ (close to 0), with 0 denoting the empty set.  For example,
$\mathcal{R}_1^2$ denotes that subdomain with $u_1 \approx 1$ and $u_2 \approx
0$, and $\mathcal{R}_0^2$ the subdomain where $u_2$ is near $1$, but $u_1$ is
away from the boundary. The middle column define the subdomain explicitly.   
The right column gives the differential-algebraic system that approximates
Eq.~\eqref{EquationForToggleSwitch}  within the given subdomain. }
\label{table:2d}
}
\end{center}
\end{table}
%

%{\bf You should try using displaystyle, since the fractions that appear in the
%right column seem very small.} \AK{It does not look very good with big fractions
%in the boxes.}

%
\begin{figure}[t]
\begin{center}
\includegraphics[scale = .4]{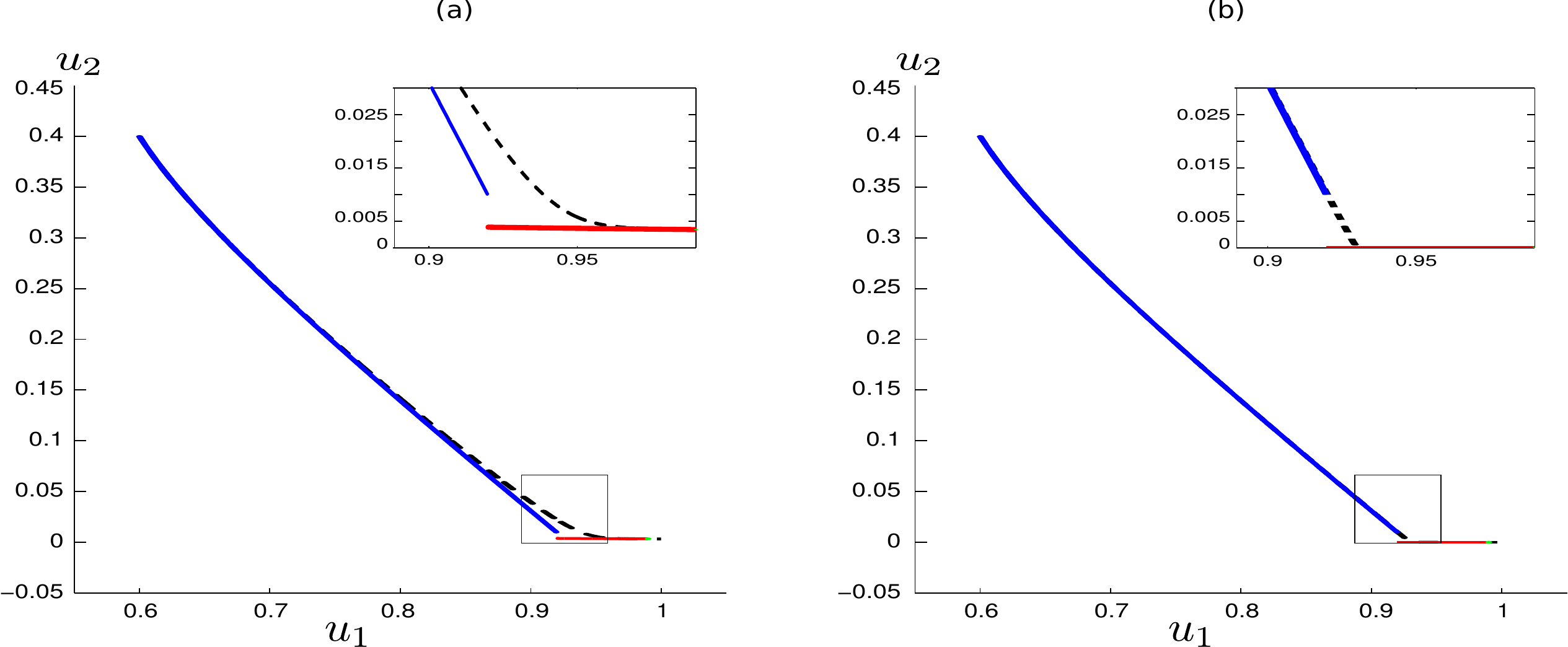}
\parbox{.9\textwidth}{
\caption{\footnotesize
Comparison of the numerical solution of Eq.~(\ref{EquationForToggleSwitch})
(dashed black) and the solution of the approximate system as listed in
Table~\ref{table:2d} (solid colored) for two different values of $J$ ( We used 
$J = 10^{-2}$ in (a); and $J =10^{-4} $ in (b).).  The different colors denote
the switching behavior of the solution from one subdomain to next. We used $\d =
0.01$. Solution of the linear approximation started in the subdomain
$\mathcal{R}_0^0$ (Initial value: $u_1 = 0.6, u_2 = 0.4$), and as soon as $u_2$
became smaller than $\d$, the subdomain switched to  $\mathcal{R}_2^0$ and
driving  linear differential equation also switched accordingly. It should be
noted that the approximate solution is discontinuous. The reason is that as soon
as the solution crossed the horizontal line, $u_2 = \d$, the solution jumped
(see inset) to the manifold, described by the algebraic part of the linear
differential algebraic system prevalent in the subdomain $\mathcal{R}_2^0$.
The solution finally stopped in the  subdomain $\mathcal{R}_2^1$.  }
\label{fig:twoDimension}
}
\end{center}
\end{figure}

\subsection{A network of three mutually inhibiting elements}

The same reduction can be applied to systems of arbitrary dimension.  As an
example consider
the
\emph{repressilator}~[\cite{NovakPatakiCilibertoTyson2003,ElowitzLeibler2000}]
% \Kcomment{You also need to
%cite the original repressilator paper by Elowitz and Leibler.}
described by
\begin{eqnarray}\label{threeNodeMM}
 \frac{du_1}{dt} &=& 0.6\frac{1- u_1}{J+1-u_1} -  u_3\frac{ u_1}{J+u_1}, \notag
\\
 \frac{du_2}{dt} &=& 0.4\frac{1- u_2}{J+1-u_2} -  u_1\frac{ u_2}{J+u_2}, \\
 \frac{du_3}{dt} &=& 0.3\frac{1- u_3}{J+1-u_3} -  u_2\frac{ u_3}{J+u_3}. \notag
\end{eqnarray}
%
%{\bf I would change the notation from $x,y,z$ to $u_1, u_2, u_3$ to be
%consistent with the rest of the paper.}
The cyclic repression of the three elements in this network leads to oscillatory
solutions over a large range of values of $J$.  The domain of this system,
$[0,1]^3$, can
be divided into 27 subdomains: 1 interior,
6 faces, 12 edges, and 8 vertices.  We
can again  approximate Eq.~\eqref{threeNodeMM} with solvable
differential--algebraic equation within each  subdomain, to obtain a global
approximate solution. We demonstrate the
validity of this approximation  in  Fig.~\ref{fig:threeDimension}. Note that 
both the numerically obtained solution to  Eq.~\eqref{threeNodeMM}, and its
approximation
exhibit oscillations, and that the approximation is discontinuous.

\begin{figure}[t]
\begin{center}
\includegraphics[width = .9\textwidth]{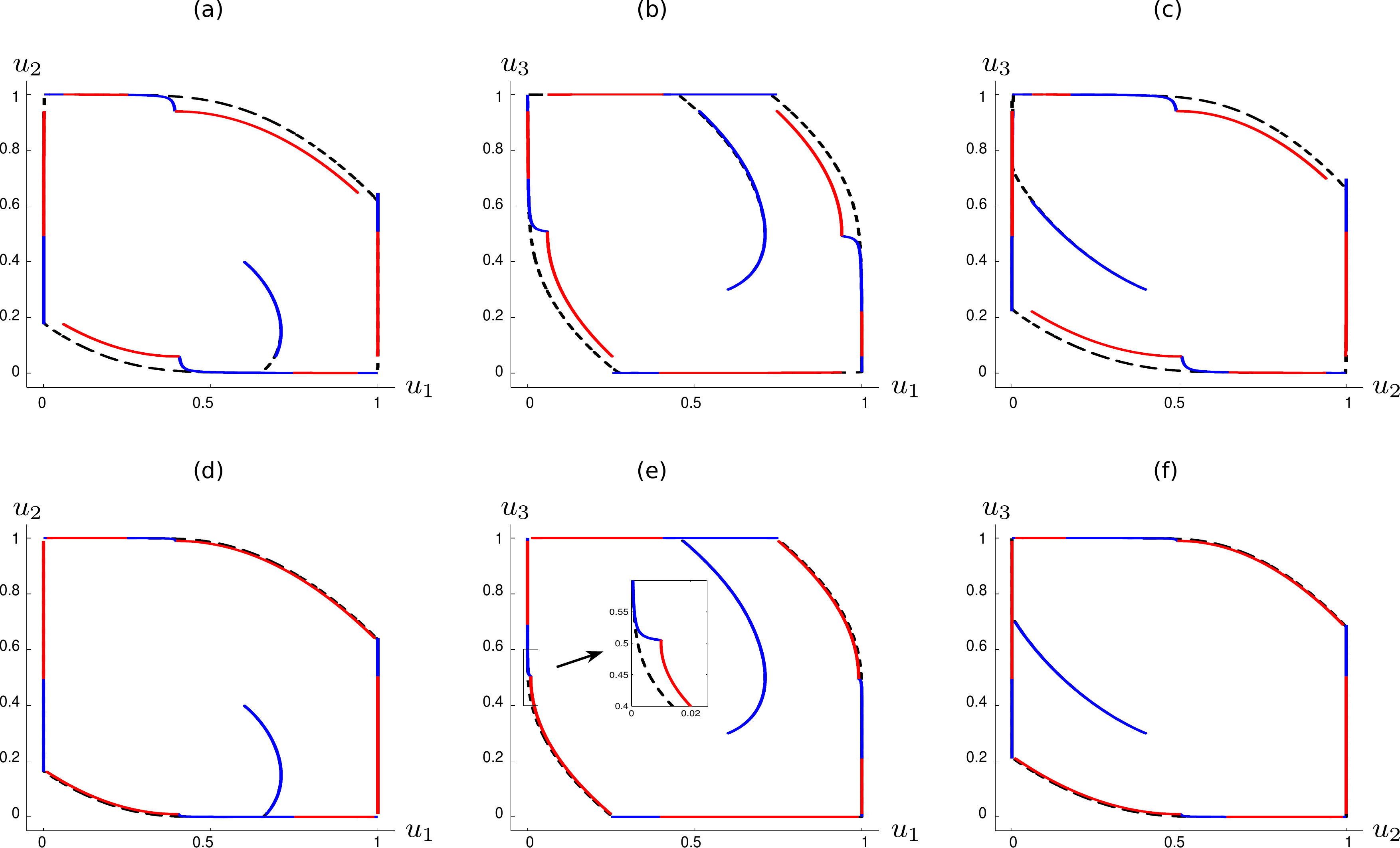}
\parbox{.9\textwidth}{
\caption{\footnotesize
Comparison of the numerical solution of Eq.~(\ref{threeNodeMM}) (dashed black)
and the solution of the approximate linear system (not explicitly provided) for
two different sets of $J$ and $\d$. For (a)-(c) $J = 10^{-2}, \d = 0.06$; for
(c)-(f)  $J = 10^{-4}, \d = 0.01$.  The approximate solution changes color when
switching between different subdomains.  Note that the approximate solution is 
discontinuous in general. The reason is that as soon as the solution enters a
new subdomain, the solution jumps (see inset) to the manifold defined by the
algebraic part of the linear differential algebraic system corresponding to the
new subdomain.
}
}
\end{center}
\label{fig:threeDimension}
\end{figure}

%{\bf You only mention the following point in the caption.  It should be made in
%the text itself:
%The reason is that as soon as the solution enters a new region, the solution
%jumps (see inset) to the manifold define by the algebraic part of the linear
%differential algebraic system corresponding to the new region.
%}

\section{General setup}\label{GeneralTheory}
The approximations described in the previous section can be extended to more
general models.  Suppose we describe the evolution of  $n$ interacting elements,
 $u_1,u_2,...,u_n$,  by
\begin{align}\label{ProblemEquation}
	\frac{du_i}{dt} = A_i
\frac{1-u_i}{J_{i}^A+1-u_i}-I_i\frac{u_i}{J_{i}^I+u_i},
\end{align}
where $J_i^A,J_i^I$ are some positive constants.  Here $A_i$ and $I_i$ are
activation/inhibition functions that capture the impact of other variables on
the evolution of $u_i$   The initial conditions are assumed to satisfy $u_i(0)
\in [0,1]$ for all $i$.

%\Kcomment{Give the activation and inhibition function names already here,
%and then specify them later.}

%\Kcomment{You need to write the preceding paragraph more clearly.  Under what
%conditions will it be true
%that the cube is invariant.  I would explain this.  Stating it in form of a
%proposition may be a good idea.}

 We assume that the activation and inhibition functions are both
affine~[\cite{de02}],
 \begin{equation} \label{activation/Inhibition}
  A_i := \sum_{j=1}^n{w_{ij}^+u_j} +b_i^+,
\quad
   I_i := \sum_{j=1}^n{w_{ij}^-u_j} +b_i^-,
 \end{equation}	
 where  we use the convention $x^+ = \text{max}\{x,0\}$ and $x^- =
\text{max}\{-x,0\}$.
The  $n\times n$ matrix, $W = [w_{ij}]$ and the  $n \times 1$ vector  $b = [\
b_1 \  b_2 \  ...\  b_n\ ]^t$ capture the connectivity and external input to the
network, respectively. In particular,  $w_{ij}$
gives  the contribution  of the $j^{\text{th}}$ variable to the growth rate of
$i^{\text{th}}$ variable. If $w_{ij} > 0 $, then $w_{ij}$  appears in the
activation function for $u_i$; and if $w_{ij}<0$ then $-w_{ij}$  appears in the
inhibition function for $u_i$.
The intensity of the external input to the $i^{\text{th}}$ element is $|b_i|$,
and it contributes to the activation or the inhibition function, depending on
whether $b_i > 0$ or $b_i < 0$,  respectively.

\begin{proposition}\label{prop:invariantCube}
If $A_i$ and $I_i$ are  positive, then the cube $[0,1]^n$ is invariant for the
dynamical system given by Eq.~(\ref{ProblemEquation}).
\end{proposition}
\begin{proof}
It will be enough to show that the vector field at any point on the boundary is
directed inward. Since, $A_i$ and $I_i$ are positive, for any $i$,
\begin{align*}
	\frac{du_i}{dt}\bigg|_{u_i = 0} = A_i\frac{1}{J_{i}^A+1} \ge 0,
\quad
\text{and}
\quad
	\frac{du_i}{dt}\bigg|_{u_i = 1} = -I_i\frac{1}{J_{i}^I+1} \le 0.
\end{align*}
\end{proof}

\section{General reduction of the model system}

To obtain a solvable reduction of Eq.~\eqref{ProblemEquation} we follow the 
procedure outlined in  Section~\ref{ExampleProblems}.  We present the result
here, and provide the mathematical justification in the next section.  For
notational convenience we consider the case $J_i^A = J_i^I = J$, with  $J$ small
 and positive.   The general case is equivalent.  Let $\d $ be some positive
number which will be used to define the thickness of the boundary layers, and
which will depend on $J$ in general.  We start with the subdivision of the
$n$-dimensional cube, $[0,1]^n$.

Let $T$ and $S$ be two disjoint subsets of  $\{1,2,...,n\}$, and let
\begin{align*}
\mathcal{R}^{T}_{S}:=
\Big\{
(u_1,u_2,...,u_n) \in [0,1]^n  \,\Big|\,
 u_{s} 		 <  \d \text{ for all }   s \in S; \quad
    &u_{t}	 	> 1-\d\text{ for all }   t \in T; \\
\text{and }
 \d \le  &u_k\le 1-\d
  \text{ for all }  \quad k  \notin S \cup T
\Big \}.
\end{align*}
We extend the convention used in Table~\ref{table:2d}, and in Eqs.~(\ref{R00})
and (\ref{R10}) so that  $\mathcal{R}^{T}_0 :=\mathcal{R}^{T}_{S}$ when $S$ is
empty;  $\mathcal{R}^{0}_S :=\mathcal{R}^{T}_{S}$ when $T$ is empty; and
 $\mathcal{R}^0_0 :=\mathcal{R}^{T}_{S}$  when $T$, $S$ are both  empty.

Within each subdomain $\mathcal{R}_S^T$ Eq.~(\ref{ProblemEquation}) can be
approximated by a
 different linear differential--algebraic system. Following the reduction from
Eq.~(\ref{EquationForToggleSwitch}) to Eq.~(\ref{edgeLinear2d}), for $i \notin S
\cup T$ we obtain the linear system
\begin{subequations}\label{equationInRST}
\begin{equation} \label{four}
 \frac{du_i}{dt} = \sum_{j = 1}^na_{ij}u_j +  b_i.
\end{equation}
For $s \in S$ one of the nonlinear terms remains and we  obtain
\begin{equation} \label{five}
\frac{du_s}{dt} = \left(\sum_{j = 1}^na_{sj}^+u_j +  b_s^+\right)-\left(\sum_{j
= 1}^na_{sj}^-u_j +  b_s^-\right)\frac{u_s}{J+u_s},
\end{equation}
while for $t \in T$ we will have
\begin{equation} \label{six}
\frac{d u_t}{dt} = \left(\sum_{j = 1}^na_{tj}^+u_j + 
b_t^+\right)\frac{1-u_t}{J+1-u_t}
-\left(\sum_{j = 1}^n a_{tj}^-u_j +  b_t^-\right).
\end{equation}
\end{subequations}
Eq.~(\ref{equationInRST}) is simpler than Eq.~(\ref{ProblemEquation}), but it is
not solvable yet.  Following the reduction from Eq.~(\ref{edgeLinear2d}) to
Eq.~(\ref{smallu}),  we now further reduce Eqs.(\ref{five}--\ref{six}).
First we use the approximations $u_s \approx 0$ and $u_t \approx 1$ in the
activation and inhibition functions appearing in Eq.~(\ref{equationInRST}).
Second, we assume
that $u_s$ for $s \in S$ and $u_t$ for $t \in T$ are in steady state.

Under these assumptions we obtain the reduction of Eq.~\eqref{ProblemEquation}
within any subdomain $\mathcal{R}_S^T$
\begin{subequations}\label{mainReduced_RST}
\begin{align}
\frac{du_i}{dt} &=  \sum_{j  \notin S \cup T}a_{ij}u_j + \sum_{j \in T}a_{ij} 
+b_i
& i \notin S \cup T;  \label{mainReduced_RSTa} \\
0 &=  \sum_{j \notin S \cup T}a_{sj}^+u_j + \sum_{t \in T}a_{st}^+
  +b_s^+ -\left( \sum_{j \notin S \cup T }a_{sj}^-u_j + \sum_{t \in T}a_{st}^- +
b_s^-\right)\frac{u_s}{J+u_s};
& s \in S, \label{mainReduced_RSTb} \\
0 &= -\left( \sum_{j \notin S \cup T }a_{tj}^+u_j + \sum_{j \in
T}a_{tj}^++b_t^+\right)\frac{1-u_t}{J+1-u_t}
+\sum_{j \notin S \cup T}a_{tj}^-u_j + \sum_{j \in T}a_{tj}^- + b_t^-,
& t \in T. \label{mainReduced_RSTc}
\end{align}
\end{subequations}

 Eq.~(\ref{mainReduced_RST}) is  solvable since Eq.~(\ref{mainReduced_RSTa}) is
decoupled from the rest, and   Eqs.(\ref{mainReduced_RSTb}) and
(\ref{mainReduced_RSTc}) are solvable for ${u}_s$ and ${u}_t$, respectively, as
functions of the solution of Eq.~(\ref{mainReduced_RSTa}).

\section{Mathematical justification}

We next justify our claim that the variables  close to the boundary
can be assumed to be in steady state. We define the following new variables to
``magnify" the boundary region.
\begin{equation} \label{seven}
 \tilde{u}_s := \frac{u_s}{J}  \quad   \text{ for }  s \in S, \text{ and } 
\qquad
 \tilde{u}_t  := \frac{1-u_t}{J} \,\,\quad \text{ for }  t \in T.
\end{equation}
Using Eq.~ (\ref{seven}) in Eq.~(\ref{equationInRST}) we get for $i \notin S
\cup T$
\begin{subequations}\label{equationInRST_scaled}
\begin{equation} \label{eight}
\frac{du_i}{dt} =  \sum_{j  \notin S \cup T}a_{ij}u_j + \sum_{j \in T}a_{ij}
  + J\left(\sum_{s \in S}a_{is}\tilde{u}_s-\sum_{t \in
T}a_{it}\tilde{u}_t\right) +b_i,
\end{equation}
and for $s \in S$,
\begin{align} \label{nine}
J\frac{d\tilde{u}_s}{dt} = & \sum_{j \notin S \cup T}a_{sj}^+u_j + \sum_{t \in
T}a_{st}^+
  +J\left(\sum_{j \in S}a_{sj}^+\tilde{u}_j-\sum_{t \in
T}a_{st}^+\tilde{u}_t\right)+b_s^+ \nonumber\\
&-\left( \sum_{j \notin S \cup T }a_{sj}^-u_j + \sum_{t \in T}a_{st}^- +
b_s^-\right)\frac{\tilde{u}_s}{1+\tilde{u}_s}
  -J\left(\sum_{j \in S}a_{sj}^+\tilde{u}_j-\sum_{t \in
T}a_{st}^+\tilde{u}_t5\right)\frac{\tilde{u}_s}{1+\tilde{u}_s},
\end{align}
and similarly, for $t \in T$,
\begin{align} \label{ten}
J\frac{d\tilde{u}_t}{dt} = &-\left( \sum_{j \notin S \cup T }a_{tj}^+u_j +
\sum_{j \in T}a_{tj}^++b_t^+\right)\frac{\tilde{u}_t}{1+\tilde{u}_t}
 -J\left(\sum_{s \in S}a_{ts}^+\tilde{u}_s-\sum_{j \in
T}a_{tj}^+\tilde{u}_j\right)\frac{\tilde{u}_t}{1+\tilde{u}_t} \nonumber \\
&+\sum_{j \notin S \cup T}a_{tj}^-u_j + \sum_{j \in T}a_{tj}^- + b_t^-
+J\left(\sum_{s \in S}a_{ts}^+\tilde{u}_s-\sum_{j \in
T}a_{tj}^+\tilde{u}_j\right).
\end{align}
\end{subequations}
%
%Here we used $J$ as one common value for all $J_s^I, J_t^A$ for notational
%convenience.

When $J$ is smallñ, we can apply Geometric Singular Perturbation Theory (GSPT) to Eq. (\ref{equationInRST_scaled}) ~[\cite{Hek2010,Kaper1998}].  The GSPT posits that, under a normal hyperbolicity condition which we will prove below, Eq.~(\ref{equationInRST_scaled}) can be further simplified by assuming that $J = 0$. 
%Note that Eq.~(\ref{equationInRST_scaled})  has the form of
%Eq.~(\ref{10may10_GSPT}), with $u_i$, $i \not \in S \cup T$ as the \emph{slow
%variable}; and $\tilde{u}_s, s \in S$ and $\tilde{u}_t, t \in T$ as the
%\emph{fast} variables. Geometric Singular Perturbation Theory~[\cite{Hek2010}]
%that we are justified in setting $J=0$ in Eq.~(\ref{equationInRST_scaled}) to
%obtain a
This yields a 
differential-algebraic system
\begin{subequations}\label{mainReduced}
\begin{align}
\frac{du_i}{dt} &=  \sum_{j  \notin S \cup T}a_{ij}u_j + \sum_{j \in T}a_{ij} 
+b_i,
& i \notin S \cup T;  \label{eleven} \\
0 &=  \sum_{j \notin S \cup T}a_{sj}^+u_j + \sum_{t \in T}a_{st}^+
  +b_s^+ -\left( \sum_{j \notin S \cup T }a_{sj}^-u_j + \sum_{t \in T}a_{st}^- +
b_s^-\right)\frac{\tilde{u}_s}{1+\tilde{u}_s},
& s \in S; \label{twelve} \\
0 &= -\left( \sum_{j \notin S \cup T }a_{tj}^+u_j + \sum_{j \in
T}a_{tj}^++b_t^+\right)\frac{\tilde{u}_t}{1+\tilde{u}_t}
+\sum_{j \notin S \cup T}a_{tj}^-u_j + \sum_{j \in T}a_{tj}^- + b_t^-,
& t \in T. \label{thirteen}
\end{align}
\end{subequations}
which is equivalent to Eq.~\eqref{mainReduced_RST} after rescaling.  This
conclusion
will be justified if the manifold defined by Eqs.~(\ref{twelve}) and
(\ref{thirteen}) is normally hyperbolic and
stable~[\cite{Fenichel1979,Kaper1998,Hek2010}].  We verify this condition next.

Let $ \hat{u} =  \{ u_{i_1},...,u_{i_m} \}$ where $\{ i_1,...,i_m \} = \{
1,2,...,n \}\backslash (S \cup T) $, be the coordinates of $u$ which are away
from the boundary, and denote the right hand side of Eq.~(\ref{twelve}) by
$F_s(\hat{u}, \tilde{u}_{i_s})$, for all  $s \in S$, so that
\begin{align*}
 F_s(\hat{u}, \tilde{u}_{i_s}) :=  \sum_{j \notin S \cup T}a_{sj}^+u_j + \sum_{t
\in T}a_{st}^+
  +b_s^+ -\left( \sum_{j \notin S \cup T }a_{sj}^-u_j + \sum_{t \in T}a_{st}^- +
b_s^-\right)\frac{\tilde{u}_s}{1+\tilde{u}_s},
\end{align*}
and
\begin{align*}
 \frac{\partial F_s}{\partial \tilde{u}_{i_s}} =  -\left( \sum_{j \notin S \cup
T }a_{sj}^-u_j + \sum_{t \in T}a_{st}^- + b_s^-\right) \left(
\frac{1}{1+\tilde{u}_s}\right)^2 ,
< 0
\end{align*}
for all $s \in S$.
Similarly, by denoting the right hand side of Eq.~(\ref{thirteen}) by
$G_t(\hat{u}, \tilde{u}_{i_t})$, for all  $t \in T$. \emph{i.e.}
\begin{align*}
 G_t(\hat{u}, \tilde{u}_{i_t}) := -\left( \sum_{j \notin S \cup T }a_{tj}^+u_j +
\sum_{j \in T}a_{tj}^++b_t^+\right)\frac{\tilde{u}_t}{1+\tilde{u}_t}
+\sum_{j \notin S \cup T}a_{tj}^-u_j + \sum_{j \in T}a_{tj}^- + b_t^-,
\end{align*}
we see that
\begin{align*}
 \frac{\partial G_t}{\partial \tilde{u}_{i_t}} = -\left( \sum_{j \notin S \cup T
}a_{tj}^+u_j + \sum_{j \in T}a_{tj}^++b_t^+\right) \left(
\frac{\tilde{u}_t}{1+\tilde{u}_t} \right)^2
< 0.
\end{align*}
Hence, the manifold defined by Eqs.~(\ref{twelve}) and (\ref{thirteen}) is
normally hyperbolic and stable.  This completes the proof that the reduction of
the non-linear system~(\ref{ProblemEquation}) to a solvable
system~(\ref{mainReduced_RST}) is justified for small $J$.

%The following notations and algebra will make this idea more clear.

%Suppose $T = \{t_1,t_2,...,t_r\}$ and $S = \{s_1,s_2,...,s_l\}$ then we define
%$\tilde{u}_T = ( \tilde{u}_{t_1}, \tilde{u}_{t_2}, \cdots ,\tilde{u}_{t_r} )$
%and  $\tilde{u}_S = ( \tilde{u}_{s_1}, \tilde{u}_{s_2}, \cdots ,\tilde{u}_{s_l}
%)$ and $v = (u_{i_1}, )$

\section{Discussion}\label{Discussion}
 A special class of non-linear differential equation was studied with non-linear interaction
terms given by Hill functions. We showed that when the Michaelis-Menten
constants are sufficiently small, the behavior of the  system is captured by 
an approximate piecewise linear systems.  This induces a natural
decomposition of the domain into a nested sequence of hypercubes, with a separate
linear--algebraic system giving an approximation in each subdomain.  We have illustrated
the theory in examples, and justified the conclusions using GSPT. 

A potential limitation in our arguments is that we have an approximation 
valid only in an asymptotic limit.  It is unknown when and how the approximation breaks down.
 Another major limitation of our analysis is
that we have not provided a systematic relationship between the thickness of the
boundary, $\d$, and the Michaelis-Menten constant, $J$. Numerical tests suggest that $J = \mathcal{O}(\d^2)$. 

%We believe that $\d \approx J$.  

%\Kcomment{Mention that other approaches are possible.  For instance, the theory
%of monotone dynamical
%systems can be used to study the dynamics of large networks.  However, it is
%restricted to networks
%with a particular structure.}

%-----------------------------------------------------------

%-----------------------------------------------------------

%\bibliography{ajitkumar22_main}%{}
%\bibliographystyle{model2-names}
\bibliographystyle{plainnat}

\end{document}